\documentclass[11pt]{article}
\usepackage[top=1in,left=1in,right=1in,bottom=1in]{geometry}

\newfont{\mycrnotice}{ptmr8t at 7pt}
\newfont{\myconfname}{ptmri8t at 7pt}
\let \originalleft \left
\let\originalright\right
\renewcommand{\left}{\mathopen{}\mathclose\bgroup\originalleft}
\renewcommand{\right}{\aftergroup\egroup\originalright}

\newcommand{\comment}[1]{}

\newcommand{\wcost}{\omega}
\newcommand{\branchfactor}{k}
\newcommand{\mb}[1]{{\mbox{\emph{#1}}}}
\newcommand{\RRAM}{ReRAM}

\usepackage{times,amsmath,amssymb,subfigure,algorithmicx,algorithm,graphicx,balance,fancybox,color}
\definecolor{ForestGreen}{rgb}{0.1333,0.5451,0.1333}
\usepackage[%letterpaper,dvips,
            colorlinks,linkcolor=ForestGreen,citecolor=ForestGreen,
            backref,
           ]{hyperref}
\usepackage[noend]{algpseudocode}
\usepackage{amsthm}

\makeatletter
\let\OldStatex\Statex
\renewcommand{\Statex}[1][3]{%
  \setlength\@tempdima{\algorithmicindent}%
  \OldStatex\hskip\dimexpr#1\@tempdima\relax}
\newcommand{\Statexi}[1][1]{%
  \setlength\@tempdima{\algorithmicindent}%
  \OldStatex\hskip\dimexpr#1\@tempdima\relax}
\newcommand{\Statexii}[1][2]{%
  \setlength\@tempdima{\algorithmicindent}%
  \OldStatex\hskip\dimexpr#1\@tempdima\relax}
\makeatother
\newcommand{\defn}[1]{\emph{\textbf{#1}}}
\newcommand{\myparagraph}[1]{\smallskip\noindent {\bf #1.}}
\newcommand{\id}[1]{\ifmmode\mathit{#1}\else\textit{#1}\fi}
\newcommand{\const}[1]{\ifmmode\mbox{\textc{#1}}\else\textsc{#1}\fi}
\newtheorem{theorem}{Theorem}[section]
\newtheorem{lemma}[theorem]{Lemma}
\newtheorem{corollary}[theorem]{Corollary}

\usepackage{microtype}

\begin{document}
\date{}

\title{Sorting with Asymmetric Read and Write Costs\footnote{The original version of this paper
is published in SPAA 2015.  In this version we fixed some minor typos and errors, and clarified some definitions.}}

\author{Guy
  E. Blelloch\\Carnegie Mellon University\\guyb@cs.cmu.edu \and
  Jeremy T. Fineman\\Georgetown
    University\\jfineman@cs.georgetown.edu \and Phillip
  B. Gibbons\\Carnegie Mellon University\\
    gibbons@cs.cmu.edu  \and Yan Gu\\Carnegie Mellon
    University\\ yan.gu@cs.cmu.edu \and Julian Shun\\Carnegie Mellon
    University\\ jshun@cs.cmu.edu}
\maketitle

\begin{abstract}

Emerging memory technologies have a significant gap between the cost,
both in time and in energy, of writing to memory versus reading from
memory.  In this paper we present models and algorithms that account
for this difference, with a focus on write-efficient sorting
algorithms.  First, we consider the PRAM model with asymmetric write
cost, and show that sorting can be performed in $O\left(n\right)$
writes, $O\left(n \log n\right)$ reads, and logarithmic depth
(parallel time).  Next, we consider a variant of the External Memory
(EM) model that charges $\wcost > 1$ for writing a block of size $B$
to the secondary memory, and present variants of three EM sorting
algorithms (multi-way mergesort, sample sort, and heapsort using
buffer trees) that asymptotically reduce the number of writes over the
original algorithms, and perform roughly $\wcost$ block reads for
every block write.  Finally, we define a variant of the Ideal-Cache
model with asymmetric write costs, and present write-efficient,
cache-oblivious parallel algorithms for sorting, FFTs, and matrix
multiplication.  Adapting prior bounds for work-stealing and
parallel-depth-first schedulers to the asymmetric setting, these yield
parallel cache complexity bounds for machines with private caches or with
a shared cache, respectively.

\end{abstract}

%\category{F.2.2}{Analysis of Algorithms and Problem Complexity}{Nonnumerical Algorithms and Problems---\emph{Sorting and searching}}
%\keywords{Sorting, asymmetric read-write costs,
%non-volatile memory, persistent memory, write-efficient, write-avoiding,
%parallel algorithms, cache-oblivious algorithms, external memory model,
%mergesort, sample sort, I/O buffer tree, FFT, matrix multiplication.}

\section{Introduction}\label{sec:intro}

Emerging nonvolatile/persistent memory (NVM) technologies such as
Phase-Change
Memory (PCM), Spin-Torque Transfer Magnetic RAM (STT-RAM), and
Memristor-based Resistive RAM (\RRAM) offer the promise of
significantly lower energy and higher density (bits per area) than
DRAM. With byte-addressability and read latencies approaching or
improving on DRAM speeds, these NVM technologies are projected to
become the dominant memory within the decade~\cite{Meena14,Yole13}, as
manufacturing improves and costs decrease.

% Observing this trend, many computer architecture, databases, and systems
% researchers have studied ways to best take advantage of
% emerging NVM~\cite{Akel11,Athanassoulis12,ChoL09,ConditNFILBC09,
% DeBrabant14,Dong08,Dong09,Dulloor14,Kim14,Koltsidas14,
% LeIMB09,QureshiKFSLA09,QureshiSR09,Qureshi12,
% SeongWL10,SchechterLSB10,Xu11,yang:iscas07,ZhouZYZ09}.

Although these NVMs could be viewed as just a layer in the memory hierarchy
that provides persistence, there is one important distinction:
\emph{Writes are significantly more costly than reads,} suffering from
higher latency, lower per-chip
bandwidth, higher energy costs, and endurance problems (a cell wears
out after $10^{8}$--$10^{12}$ writes~\cite{Meena14}).  Thus, unlike DRAM, there
is a significant (often an order of magnitude or more) asymmetry
between read and write costs~\cite{Akel11,Athanassoulis12,
Dong09,Dong08,ibm-pcm14b,Kim14,Qureshi12,Xu11}, motivating the
study of {\em write-efficient} ({\em write-limited}, {\em write-avoiding})
algorithms, which reduce the number of writes relative to existing algorithms.

\myparagraph{Related Work}
While read-write asymmetry in NVMs has been the focus of many systems
efforts~\cite{ChoL09,LeeIMB09,yang:iscas07,ZhouZYZ09,ZWT13}, there have been
very few papers addressing this from the algorithmic perspective.
Reducing the number of writes has long been a goal in disk arrays,
distributed systems, cache-coherent multiprocessors, and the like,
but that work has not focused on NVMs and the solutions are not suitable
for their properties.  Several
papers~\cite{BT06,Eppstein14,Gal05,nath:vldbj10,ParkS09,Viglas14} have
looked at
read-write asymmetries in the context of NAND flash memory. This work
has focused on (i) the fact that in NAND flash chips, bits can only be
cleared by incurring the overhead of erasing a large block of memory,
and/or (ii) avoiding individual cell wear out.  Eppstein et
al.~\cite{Eppstein14}, for
example, presented a novel cuckoo hashing algorithm that chooses where
to insert/reinsert each item based on the number of writes to date for
the cells it hashes to.  Emerging NVMs, in contrast, do not suffer
from (i) because they can write arbitrary bytes in-place.  As for
(ii), we choose not to focus on wear out (beyond reducing the total
number of writes to all of memory) because system software (e.g., the garbage
collector, virtual memory manager, or virtual machine hypervisor) can
readily balance application writes across the physical memory over the long
time horizons (many orders of magnitude longer than NAND Flash) before an
individual cell would wear out from too many writes to it.

A few prior works~\cite{Chen11,Viglas14,Viglas12} have looked at
algorithms for asymmetric read-write costs in emerging NVMs, in the
context of databases. Chen et al.~\cite{Chen11} presented analytical
formulas for PCM latency and energy, as well as algorithms for
B-trees and hash joins that are tuned for PCM. For example, their
B-tree variant does not sort the keys in a leaf node nor repack a leaf
after a deleted key, thereby avoiding the write cost of sorting and
repacking, at the expense of additional reads when
searching. Similarly, Viglas~\cite{Viglas12} traded off fewer writes
for additional reads by rebalancing a B$^+$-tree only if the cost of
rebalancing has been amortized.  Viglas~\cite{Viglas14} presented
several ``write-limited'' sorting and join algorithms within the
context of database query processing.

\myparagraph{Our Results}
In this paper, we seek to systematically study algorithms under
asymmetric read and write costs. We consider natural extensions to the
RAM/PRAM models, to the External Memory model, and to the Ideal-Cache
model to incorporate an integer cost, $\wcost > 1$, for writes
relative to reads.
% The latter model is used to design low-depth
% cache-oblivious algorithms, which can be mapped to private or shared
% cache machines with provably good bounds.
We focus primarily on sorting algorithms, given their fundamental role
in computing and as building blocks for other problems, especially in
the External Memory and Parallel settings---but we also consider
additional problems.

We first
observe that in the RAM model, it is well known that sorting by
inserting each key into a balanced search tree requires only
$O\left(n\right)$ writes with no increase in reads
($O\left(n \log n\right)$).  Applying this idea to a carefully-tuned
sort for the asymmetric CRCW PRAM yields a parallel sort with
$O\left(n\right)$ writes, $O\left(n \log n\right)$ reads and
$O\left(\wcost{} \log n\right)$ depth (with high probability\footnote{\emph{With high probability (w.h.p.)} means with probability $1 - n^{-c}$, for a
constant $c$.}).
% that is at most linear in the constant in the Big-O notation.}).

Next, we consider an Asymmetric External Memory (AEM) model, which has
a small primary memory (cache) of size $M$ and transfers data in blocks
of size $B$ to (at a
cost of $\wcost$) and from (at unit cost) an unbounded external memory.
We show that three asymptotically-optimal EM
sorting algorithms can each be adapted to the AEM with reduced write
costs.  First, following~\cite{ParkS09,Viglas14}, we adapt multi-way
mergesort by merging $O(\wcost) M/B$ sorted runs at a time (instead of
$M/B$ as in the original EM version).  This change saves writes by
reducing the depth of the recursion.  Each merge makes $O(\wcost)$ passes
over the runs, using an in-memory heap to extract values for the
output run for the pass.  Our algorithm and analysis is somewhat
simpler than~\cite{ParkS09,Viglas14}.  Second, we present an AEM
sample sort algorithm that uses $O(\wcost) M/B$ splitters at each level
of recursion (instead of $M/B$ in the original EM version). Again, the
challenge is to both find the splitters and partition using them while
incurring only $O\left(N/B\right)$ writes across each level of
recursion.  We also show how this algorithm can be parallelized to run
with linear speedup on the Asymmetric Private-Cache model
(Section~\ref{sec:prelims}) with $p = n/M$ processors.  Finally, our
third sorting algorithm is an AEM heapsort using a buffer-tree-based
priority queue.  Compared to the original EM algorithm, both our
buffer-tree nodes and the number of elements stored outside the buffer
tree are larger by a factor of $O(\wcost{})$, which adds nontrivial
changes to the data structure.  All three sorting algorithms have the
same asymptotic complexity on the AEM.

Finally, we define an Asymmetric Ideal-Cache model, which is similar
to the AEM model in terms of $M$ and $B$ and having asymmetric read/write
costs, but uses an asymmetric ideal replacement policy instead of
algorithm-specified transfers.  We extend important results for the
Ideal-Cache model and thread schedulers to the asymmetric
case---namely, the Asymmetric Ideal-Cache can be (constant factor)
approximated by an asymmetric-LRU cache, and it can be used in
conjunction with a work-stealing (parallel-depth-first) scheduler to
obtain good parallel cache complexity bounds for machines with private caches
(a shared cache, respectively).  We use this model to design write-efficient
cache-oblivious algorithms for sorting, Fast Fourier Transform, and
matrix multiplication.
%---which can be mapped to private or shared cache
%machines with provably good bounds, using a work-stealing or
%parallel-depth-first scheduler.
Our sorting algorithm is adapted
from~\cite{BGS10}, and again deals with the challenges of reducing the
number of writes.  All three algorithms use
$\Theta\left(\wcost\right)$ times more reads than writes and have good parallelism.

\section{Preliminaries and Models}\label{sec:prelims}

This section presents background material on NVMs and models, as well as
new (asymmetric cost) models and results relating models.
We first consider models whose parallelism is in the parallel transfer
of the data in a larger block, then consider models with parallel
processors.

\myparagraph{Emerging NVMs} While DRAM stores data in capacitors that
typically require refreshing every few milliseconds,
and hence must be continuously powered, emerging NVM
technologies store data as ``states'' of the given material that
require no external power to retain.  Energy is required only to read
the cell or change its value (i.e., its state).  While there is no
significant cost difference between reading and writing DRAM (each
DRAM read of a location not currently buffered requires a write of
the DRAM row being evicted, and hence is also a write),
emerging NVMs such as Phase-Change Memory (PCM), Spin-Torque
Transfer Magnetic RAM (STT-RAM), and Memristor-based Resistive RAM
(\RRAM) each incur significantly higher cost for writing than reading.
This large gap seems fundamental to the technologies themselves: to
change the physical state of a material requires relatively
significant energy for a sufficient duration, whereas reading the
current state can be done quickly and, to ensure the state is left
unchanged, with low energy.  An STT-RAM cell, for example, can be read
in 0.14 $ns$ but uses a 10 $ns$ writing pulse duration, using roughly
$10^{-15}$ joules to read versus $10^{-12}$ joules to
write~\cite{Dong08} (these are the raw numbers at the materials
level).  A Memristor \RRAM{} cell uses a 100 $ns$ write pulse duration, and
an 8MB Memrister \RRAM{} chip is projected to have reads with 1.7 $ns$
latency and 0.2 $nJ$ energy versus writes with 200 $ns$ latency and 25 $nJ$
energy~\cite{Xu11}---over two orders of magnitude differences in latency
and energy.  PCM is the most mature of the three technologies, and
early generations are already available as I/O devices.  A recent
paper~\cite{Kim14} reported 6.7 $\mu s$ latency for a 4KB read and
128 $\mu s$ latency for a 4KB write.  Another reported that the
sector I/O latency and bandwidth for random 512B writes was a factor
of 15 worse than for reads~\cite{ibm-pcm14b}.  As a future memory/cache
replacement, a 512Mb PCM memory chip is projected to have 16 $ns$ byte
reads versus 416 $ns$ byte writes, and writes to a 16MB PCM L3 cache
are projected to be up to 40 times slower and use 17 times more energy
than reads~\cite{Dong09}.  While these numbers are speculative and subject
to change as the new technologies emerge over time, there seems to be
sufficient evidence that writes will be considerably more costly than
reads in these NVMs.

\myparagraph{Sorting} The sorting problem we consider is the standard
comparison based sorting with $n$ records each containing a key.  We
assume the input is in an unsorted array, and the output needs to be
placed into a sorted array.  Without loss of generality,
we assume the keys are unique (a position index can always be added to make them unique).

\myparagraph{The Asymmetric RAM model} This is the standard RAM model
but with a cost $\wcost > 1$ for writes, while reads are still unit cost.

\myparagraph{The (Asymmetric) External Memory model} The widely studied
External Memory (EM) model~\cite{AggarwalV88} (also called I/O model,
Disk Model and Disk Access Model) assumes a two level memory hierarchy
with a fixed size primary memory (cache) of size $M$ and a secondary
memory of unbounded size.  Both are partitioned into blocks of size
$B$.  Standard RAM instructions can be used within the primary memory,
and in addition the model has two special \emph{memory transfer}
instructions: a \emph{read} transfers (alternatively, copies) an
arbitrary block from secondary memory to primary memory, and a
\emph{write} transfers an arbitrary block from primary to secondary
memory.  The I/O complexity of an algorithm is the total number of
memory transfers.  Sorting $n$ records can be performed in the EM model with I/O
complexity
\begin{equation}
\label{eqn:sortbound}
\Theta\left(\frac{n}{B} \log_{\frac{M}{B}} \frac{n}{B}\right)
\end{equation}
This is both an upper and lower bound~\cite{AggarwalV88}.  The upper
bound can be achieved with at least three different algorithms, a
multi-way mergesort~\cite{AggarwalV88}, a distribution
sort~\cite{AggarwalV88}, and a priority-queue (heap) sort based on
buffer trees~\cite{Arge03}.

The \emph{Asymmetric External Memory} (AEM) model simply adds a
parameter $\wcost$ to the EM model, and charges this for each write of a block.
Reading a block still has unit cost.

Throughout the paper, we assume that $M$ and $B$ are measured in terms
of the number of data objects.  If we are sorting, for example, it is
the number records.  We assume that the memory has an extra $O\left(\log
M\right)$ locations just for storing a stack to compute with.

\myparagraph{The (Asymmetric) Ideal-Cache model} The Ideal-Cache
model~\cite{Frigo99} is a variant of the EM model.  The machine model
is still organized in the same way with two memories each partitioned
into blocks, but there are no explicit memory transfer instructions.
Instead all addressable memory is in the secondary memory, but any
subset of up to $M/B$ of the blocks can have a copy resident in the
primary memory (cache).  Any reference to a resident block is a
\emph{cache hit} and is free.  Any reference to a word in a block that
is not resident is a \emph{cache miss} and requires a memory transfer
from the secondary memory.  The cache miss can replace a block in the 
cache with the loaded block, which might require \emph{evicting} a cache block.
The model makes the \emph{tall cache assumption} where $M = \Omega(B^2)$,
which is easily met in practice.
The I/O or \emph{cache complexity} of an algorithm is
the number of cache misses.  An optimal (offline) cache eviction
policy is assumed---i.e., one that minimizes the I/O complexity.  It is
well known that the optimal policy can be approximated using the
online least recently used (LRU) policy at a cost of at most doubling
the number of misses, and doubling the cache size~\cite{Sleator85}.

The main purpose of the Ideal-Cache model is for the design of
\emph{cache-oblivious algorithms}.  These are algorithms that do not
use the parameters $M$ and $B$ in their design, but for which one can
still derive effective bounds on I/O complexity.  This has the
advantage that the algorithms work well for any cache sizes on any
cache hierarchies.  The I/O complexity of cache-oblivious sorting is
asymptotically the same as for the EM model.

We define the \emph{Asymmetric Ideal-Cache} model by distinguishing
reads from writes, as follows.  A cache block is \emph{dirty} if the
version in the cache has been modified since it was brought into the
cache, and \emph{clean} otherwise.  When a cache miss evicts a clean
block the cost is $1$, but when evicting a dirty block the cost is $1
+ \wcost$, $1$ for the read and $\wcost$ for the write.  Again, we
assume an ideal offline cache replacement policy---i.e., minimizing the
total I/O cost.  Under this model we note that the LRU policy is no
longer 2-competitive.  However, the following variant is competitive
within a constant factor.  The idea is to separately maintain two
equal-sized pools of blocks in the cache (primary memory), a read pool
and a write pool.  When reading a location, (i) if its block is in the
read pool we just read the value, (ii) if it is in the write pool we
copy the block to the read pool, or (iii) if it is in neither, we read
the block from secondary memory into the read pool.  In the latter two
cases we evict the LRU block from the read pool if it is full, with
cost 1.  The rules for the write pool are symmetric when writing to a
memory location, but the eviction has cost $\wcost+1$ because the
block is dirty.  We call this the read-write LRU policy.  This policy
is competitive with the optimal offline policy:

\begin{lemma}
For any sequence $S$ of instructions, if it has cost $Q_I\left(S\right)$ on
the Asymmetric Ideal-Cache model with cache size $M_I$,
then it will have cost
\[Q_L\left(S\right) \leq \frac{M_L}{\left(M_L - M_I\right)}Q_I\left(S\right) + (1 + \wcost) M_I/B\] on
an asymmetric cache with read-write LRU policy and cache sizes (read
and write pools) $M_L$.
\end{lemma}

\begin{proof} Partition the sequence of instructions into
regions that contain memory reads to exactly $M_L/B$ distinct memory
blocks each (except perhaps the last).  Each region will require at most $M_L/B$
misses under LRU.  Each will also require at least $(M_L - M_I)/B$ cache
misses on the ideal cache since at most $M_I/B$ blocks can be in the
cache at the start of the region.  The same argument can be made for
writes, but in this case each operation involves evicting a dirty block.
The $(1 + \wcost) M_I/B$ is for the last region.  
To account for the last region, in the worst case at the start of the
last write region the ideal cache starts with $M_I/B$ blocks which get
written to, while the LRU starts with none of those blocks.  The LRU therefore
invokes an addition $M_I/B$ write misses each costing $1 + \wcost$ (1
for the load and $\wcost$ for the eviction).   
Note that if the cache starts empty then we do not have to add this term
since an equal amount will be saved in the first round.
\end{proof}

\myparagraph{The Asymmetric PRAM model} In the \emph{Asymmetric PRAM},
the standard PRAM is augmented such that each write costs $\wcost$ and
all other instructions cost $1$.  In this paper we analyze algorithms
in terms of work (total cost of the operations) and depth (parallel
time using an unbounded number of processors).  If we have depth
$d\left(n\right)$ and separate the work into $w\left(n\right)$ writes
and $r\left(n\right)$ other instructions, then the time on $p$
processors is bounded by:
\[T\left(n,p\right) = O\left(\frac{\wcost w\left(n\right) +
    r\left(n\right)}{p} + d\left(n\right)\right)\]
using Brent's theorem~\cite{JaJa92}.  This bound assumes that work can
be allocated to processors efficiently.  We allow for concurrent reads
and writes (CRCW), and for concurrent writes we assume an arbitrary
write takes effect.  Note that a parallel algorithm that require
$O\left(D\right)$ depth in the PRAM model requires
$O\left(\wcost D\right)$ depth in the asymmetric PRAM model to account
for the fact that writes are $\wcost$ times more expensive than reads.

\myparagraph{The Asymmetric Private-Cache model} In the
\emph{Asymmetric Private-Cache} model (a variant of the Private-Cache
model~\cite{Acar02,Arge2008}), each processor has its own primary
memory of size $M$, and all processors share a secondary memory.  We
allow concurrent reads but do not use concurrent writes.  As in the
AEM model, transfers are in blocks of size $B$ and transfers to the
shared memory cost $\wcost$.

\myparagraph{The (Asymmetric) Low-depth Cache-Oblivious Paradigm} The
final model that we consider is based on developing low-depth
cache-oblivious algorithms~\cite{BGS10}.  In the model algorithms are
defined as nested parallel computations based on parallel loops,
possibly nested (this is a generalization of a PRAM).  The depth of
the computation is the longest chain of dependences---i.e., the depth
of a sequential strand of computation is its sequential cost, and the
depth of a parallel loop is the maximum of the depths of its iterates.
The computation has a natural sequential order by converting each
parallel loop to a sequential loop.  The cache complexity can be
analyzed on the Ideal-Cache model under this sequential order.

Using known scheduling results the depth and sequential cache
complexity of a computation are sufficient for deriving bounds on
parallel cache complexity.  In particular, let $D$ be the depth and
$Q_1$ be the sequential cache complexity.  Then for a $p$-processor
shared-memory machine with private caches (each processor has its own
cache) using a work-stealing scheduler, the total number of misses
$Q_p$ across all processors is at most $Q_1 + O\left(pDM/B\right)$
with high probability~\cite{Acar02}.  For a $p$-processor shared-memory
machine with a shared cache
of size $M + p B D$ using a parallel-depth-first (PDF) scheduler,
$Q_p \leq Q_1$~\cite{BlGi04}.  These bounds can be extended to
multi-level hierarchies of private or shared caches,
respectively~\cite{BGS10}.  Thus, algorithms with low depth have good
parallel cache complexity.

Our asymmetric variant of the low-depth cache-oblivious paradigm
simply accounts for $\wcost$ in the depth and uses the Asymmetric
Ideal-Cache model for sequential cache complexity.  We observe that
the above scheduler bounds readily extend to this asymmetric setting.
The $O\left(pDM/B\right)$ bound on the additional cache misses under
work-stealing arises from an $O\left(pD\right)$ bound on the number of steals
and the observation that each steal requires the stealer to incur
$O\left(M/B\right)$ misses to ``warm up'' its cache.  Pessimistically, we will
charge $2M/B$ writes (and reads) for each steal, because each line may
be dirty and need writing back before the stealer can read it into its
cache and, once the stealer has completed the stolen work (reached the
join corresponding to the fork that spawned the stolen work), the
contents of its cache may need to be written back.  Therefore for
private caches we have $Q_P \leq Q_1 + O\left(pkDM/B\right)$.  The PDF bounds
extend because there are no additional cache misses and hence no
additional reads or writes.

\section{Sorting on RAM/PRAM}\label{sec:ram-models}

The number of writes on an asymmetric RAM can be bound for a variety
of algorithms and data structures using known techniques.  For example,
there has been significant research on maintaining balanced search
trees such that every insertion and deletion only requires a constant
number of rotations (see e.g.,~\cite{Ottmann} and references within).
While the motivation for that work is that for certain data structures
rotations can be asymptotically more expensive than visiting a node
(e.g., if each node of a tree maintains a secondary set of keys),
the results apply directly to improving bounds on the
asymmetric RAM.
Sorting can be done by inserting $n$ records into a balanced search tree
data structure, and then reading
them off in order.  This requires $O\left(n \log n\right)$ reads
and $O\left(n\right)$ writes, for total cost $O\left(n \left(\wcost +
\log n\right)\right)$.  Similarly, we can maintain priority queues
(insert and delete-min) and comparison-based dictionaries (insert,
delete and search) in $O\left(1\right)$ writes per operation.

\begin{algorithm}[!tp]
\caption{\textsc{Asymmetric-PRAM Sort}}
\label{alg:pramsort}

\textbf{Input:} An array of records $A$ of length $n$

\begin{algorithmic}[1]

\State {Select a sample $S$ from $A$ independently at random with
  per-record probability $1/\log n$, and sort the sample.}

\State {Use every $(\log n)$-th element in the sorted $S$ as splitters, and for
  each of the about $n/\log^2n$ buckets defined by the splitters
  allocate an array of size $c\log^2 n$.}

\State {In parallel locate each record's bucket using a binary search on
the splitters.}

\State {In parallel insert the records into their buckets by
  repeatedly trying a random position within the associated array and
  attempting to insert if empty.}

\State {Pack out all empty cells in the arrays and concatenate all arrays.
// Step 6 is an optional step used to obtain $O(\wcost\log n)$ depth}

%\For {round $r\leftarrow 1$ to 2} // optional, to get the $O(k\log n)$ depth

\State {{\bf For} round $r\leftarrow 1$ to 2 {\bf do}}

\Statexi {{\bf for each} array $A'$ generated in previous round}
\Statexii {Deterministically select $|A'|^{1/3}-1$ samples as splitters}
\Statexii {and apply integer sort on the bucket number to partition}
\Statexii {$A'$ into $|A'|^{1/3}$ sub-arrays.}

%\EndFor

\State {{\bf For each} subarray apply the asymmetric RAM sort.}

\State {Return the sorted array.}

\end{algorithmic}
\end{algorithm}

We now consider how to sort on an asymmetric CRCW PRAM (arbitrary
write).  Algorithm~\ref{alg:pramsort} outlines a sample sort (with
over-sampling) that does $O\left(n\log n\right)$ reads and
$O\left(n\right)$ writes and has depth $O\left(\wcost \log n\right)$.
It is similar to other sample sorts~\cite{Blelloch91,Frazer70,JaJa92}.
%Without loss of generality, we assume the records have unique keys (no
%two are equal).
We consider each step in more detail and analyze its cost.

Step 1 can use Cole's parallel mergesort~\cite{Cole1988} requiring
$O\left(n\right)$ reads and writes w.h.p.~(because the sample is size
$\Theta(n/\log n)$ w.h.p.), and $O\left(\wcost\log n\right)$ depth.  In step 2 for
sufficiently large $c$, w.h.p.~all arrays will have at
least twice as many slots as there are records belonging to the
associated
bucket~\cite{Blelloch91}.  The cost of step 2 is a lower-order term.
Step 3 requires $O(n \log n)$ reads, $O(n)$ writes and $O(\wcost +
\log n)$ depth for the binary searches and writing the resulting
bucket numbers.  Step 4 is an instance of the so-called placement
problem (see~\cite{RR89,Reif99}).  This can be implementing by having
each record select a random location within the array associated with
its bucket and if empty, attempting to insert the record at that
location.  This is repeated if unsuccessful.  Since multiple records
might try the same location at the same time, each record needs to
check if it was successfully inserted.  The expected number of tries
per record is constant.  Also, if the records are partitioned into
groups of size $\log n$ and processed sequentially within the group
and in parallel across groups, then w.h.p.~no group
will require more than $O(\log n)$ tries across all of its
records~\cite{RR89}.   Therefore, w.h.p., the number of reads and writes for
this step are $O(n)$ and the depth is $O\left(\wcost\log n\right)$.
Step 5 can be done with a prefix sum,
requiring a linear number of reads and writes, and $O\left(\wcost\log
n\right)$ depth.  At this point we could apply the asymmetric RAM sort
to each bucket giving a total of $O\left(n\log n\right)$ reads,
$O\left(n\right)$ writes and a depth of
$O\left(\wcost\log^2n+\log^2n\log\log n\right)$ w.h.p.~(the first term
for the writes and second term for the reads).

We can reduce the depth to $O\left(\wcost\log n\right)$ by
further deterministically sampling inside each bucket (step 6) using the
following lemma:

\begin{lemma}\label{lemma:partition}
We can partition $m$ records into $m^{1/3}$ buckets
$M_1,\ldots,M_{m^{1/3}}$ such that for any $i$ and $j$ where $i < j$
all records in $M_i$ are less than all records in $M_j$, and for all
$i$, $|M_i| < m^{2/3}\log m$. The process requires $O\left(m\log m\right)$ reads,
$O\left(m\right)$ writes, and $O\left(\wcost\sqrt{m}\right)$ depth.
\end{lemma}
\begin{proof}
%% The procedure is similar to the deterministic sampling procedure done
%% in~\cite{BGS10}.
We first split the $m$ records into groups of size $m^{1/3}$ and sort
each group with the RAM sort. This takes $O\left(m\log m\right)$
reads, $O\left(m\right)$ writes and $O\left(\wcost m^{1/3}\log
m\right)$ depth. Then for each sorted group, we place every $\log
m$'th record into a sample. Now we sort the sample of size $m/\log m$
using Cole's mergesort, and use the result as splitters to partition
the remaining records into buckets. Finally, we place the records into
their respective buckets by integer sorting the records based on their
bucket number.  This can be done with a parallel radix sort in a
linear number of reads/writes and $O\left(\wcost\sqrt{m}\right)$
depth~\cite{RR89}.

To show that the largest bucket has size at most $m^{2/3}\log m$,
note that in each bucket, we can pick at most $\log m$ consecutive
records from each of the $m^{2/3}$ groups without picking a
splitter. Otherwise there will be a splitter in the bucket, which is a
contradiction.
\end{proof}

%% Note that we can reduce the depth to $O(k\log^2n)$ w.h.p. by sampling
%% inside each bucket to create sub-buckets that are no larger than
%% $O(\log^2n/\log\log n)$ w.h.p. In particular, for the buckets larger
%% than $O(\log^2n/\log\log n)$, we pick $O(\log^2n/(\log\log n)^2)$
%% random elements and partition the rest of the elements into
%% sub-buckets. Again, this can be done using binary search and
%% concurrent writing using $O(n\log n)$ reads, $O(n)$ writes and
%% $O(\log^2 n)$ depth.  The maximum bucket size is $o(\log^2n/\log\log
%% n)$ w.h.p. (the expected bucket size is $O((\log\log n)^2)$ so we use
%% a Chernoff bound~\cite{MotwaniR95} which says that the probability of
%% being a factor of $O(\log n/(\log\log n)^2)$ larger than the
%% expectation has probability $O(1/n^c)$ for some constant $c>0$). At
%% this point we can apply the RAM sort to each sub-bucket in
%% $O(k\log^2n/\log\log n+\log^2n)$ depth w.h.p. This gives the following
%% theorem:

Step 6 applies two iterations of Lemma~\ref{lemma:partition} to each
bucket to partition it into sub-buckets. For an initial bucket of size
$m$, this process will create sub-buckets of at most size
$O\left(m^{4/9}\log^{5/3}m\right)$.  Plugging in $m = O\left(\log^2
n\right)$ gives us that the largest sub-bucket is of size
$O\left(\log^{8/9} n\left(\log\log n\right)^{5/3}\right)$.  We can now
apply the RAM sort to each bucket in $O\left(\wcost\log n\right)$
depth.  This gives us the following theorem.

\begin{theorem}
Sorting $n$ records can be performed using $O(n$ $\log n)$ reads,
$O\left(n\right)$ writes, and in $O\left(\wcost\log n\right)$ depth
w.h.p.~on the Asymmetric CRCW PRAM.
\end{theorem}
This implies \[T\left(n\right) = O\left(\frac{n \log n + \wcost n}{p}
+ \wcost \log n\right)\] time.  Allocating work to processors is
outlined above or described in the cited references.  In the standard
PRAM model, the depth of our algorithm matches that of the best PRAM
sorting algorithm~\cite{Cole1988}, although ours is randomized and requires
the CRCW model.   We leave it open whether the same bounds can be met
deterministically and on a PRAM without concurrent writes.

\section{External Memory Sorting}\label{sec:em-sorting}

In this section, we present sorting algorithms for the Asymmetric External Memory
model.  We show how the three
approaches for EM sorting---mergesort, sample sort, and heapsort (using
buffer trees)---can each be adapted to the asymmetric case.

In each case we trade off a factor of $\branchfactor=O(\wcost)$ additional reads for a
larger branching factor ($\branchfactor{}M/B$ instead of $M/B$), hence reducing the
number of rounds.  It is interesting that the same general approach
works for all three types of sorting.  The first algorithm, the
mergesort, has been described elsewhere~\cite{ParkS09} although in a
different model (their model is specific to NAND flash memory and has different
sized blocks for reading and writing, among other differences).  Our
parameters are therefore different, and our analysis is new.  To
the best of our knowledge, our other two algorithms are new.

\subsection{Mergesort}\label{sec:io-merge-sort}

We use an $l$-way mergesort---i.e., a balanced tree of merges
with each merge taking in $l$ sorted arrays and outputting one sorted
array consisting of all records from the input.  We assume that once
the input is small enough a different sort (the \emph{base case}) is
applied.  For $l = M/B$ and a base case of $n \leq M$ (using any sort
since it fits in memory), we have the standard EM mergesort.  With
these settings there are $\log_{M/B} \left(n/M\right)$ levels of
recursion, plus the base case, each costing
$O\left({n}/{B}\right)$ memory operations.  This gives the
well-known overall bound from
Equation~\ref{eqn:sortbound}~\cite{AggarwalV88}.
%\[\Theta\left(\frac{n}{B} \log_{\frac{M}{B}} \frac{n}{B}\right).\]

To modify the algorithm for the asymmetric case, we increase the
branching factor and the base case by a factor of $\branchfactor$ where $1\leq \branchfactor\leq \wcost$ (i.e.  $l =
\branchfactor {}M /B$ and a base case of $n \leq \branchfactor M$).  This means that
it is no longer possible to keep the base case in the primary memory,
nor one block for each of the input arrays during a merge.  The
modified algorithm is described in Algorithm~\ref{alg:mergesort}.
% For simplicity we assume keys are unique, although it is not hard to
% generalize.

\begin{algorithm}[!t]
\caption{\textsc{AEM-Mergesort}}
\label{alg:mergesort}

\textbf{Input:} An array $A$ of records of length $n$

\begin{algorithmic}[1]

\If {$|A|\leq \branchfactor M$} ~~~~// base case
\State { Sort $A$ using $\branchfactor |A|/B$ reads and $|A|/B$ writes, and return. }

%\State \Return.
\EndIf

\State {Evenly partition $A$ into $l=\branchfactor M/B$ subarrays
  $A_1,\ldots,A_{l}$ (at the granularity of blocks) and
  recursively apply \textsc{AEM-Mergesort} to each.}

\State {{\bf Initialize Merge.} Initialize an empty output array $O$,
  a load buffer and an empty store buffer each of size $B$,
  an empty priority queue $Q$ of size $M$,
  an array of pointers $I_1, \ldots, I_{l}$ that point to
  the start of each sorted subarray,
  $c=0$, and $\mb{lastV}=-\infty$.
  Associated with $Q$ is $Q.\mb{max}$, which holds the maximum element in
  $Q$ if $Q$ is full, and $+\infty$ otherwise.}

\While {$c<|A|$} %~~~~()
\label{line:while}

\For {$i\leftarrow 1$ to $l$}
\label{line:for}

\State {\textsc{Process-Block}($i$).}

\EndFor

\While {$Q$ is not empty}
\label{line:while2}
\State {$e\leftarrow Q.\mb{deleteMin}$.}

\State {Write $e$ to the store buffer, $c\leftarrow c+1$.}

\State {If the store buffer is full, flush it to $O$ and update $\mb{lastV}$.}

\If {$e$ is marked as last record in its subarray block}
\State {$i = e.\mb{subarray}$.}
\State {Increment $I_i$ to point to next block in subarray $i$.}
\State {\textsc{Process-Block}($i$).}

\EndIf

\EndWhile

\EndWhile

\State {$A\leftarrow O$.~~// Logically, don't actually copy}

\medskip

\Function{Process-Block}{subarray $i$}

\State {\textbf{If} $I_i$ points to the end of the subarray \textbf{then} return.}

%\EndIf

\State {Read the block $I_i$ into the load buffer.}

\ForAll {records $e$ in the block}

\If {$e.key$ is in the range $\left(\mb{lastV},Q.\mb{max}\right)$}

\State {If $Q$ is full, eject $Q.\mb{max}$.}

\State {Insert $e$ into $Q$, and mark if last record in block. }

\EndIf

\EndFor

\EndFunction

\end{algorithmic}
\end{algorithm}

Each merge proceeds in a sequence of rounds, where a round is one
iteration of the {\bf while} loop starting on line~\ref{line:while}.  During
each round we maintain a priority queue within the primary memory.
Because operations within the primary memory are free in the model, this
can just be kept as a sorted array of records, or even unsorted,
although a balanced search tree can be a feasible solution in practice.  Each
round consists of two phases. The first phase (the {\bf for} loop on
line~\ref{line:for}) considers each of the $l$ input subarrays in turn,
loading the current block for the subarray into the load buffer,
and then inserting each record $e$ from the block into
the priority queue if not already written to the output (i.e. $e.key >
\mb{lastV}$), and if smaller than the maximum in the queue (i.e. $e.key <
Q.\mb{max}$).  This might bump an existing element out of the queue.  Also,
if a record is the last in its block then it is marked and tagged with its
subarray number.

The second phase (the {\bf while} loop starting on line~\ref{line:while2})
starts writing the priority queue to the output one block at a time.
Whenever reaching a record that is marked as the last in its block,
the algorithm increments the pointer to the corresponding subarray
and processes the next block in the subarray.
We repeat the rounds until all records from
all subarrays have been processed.

To account for the space for the pointers $I=I_1,\ldots,I_l$, let $\alpha =
(\log n)/{s}$, where $s$ is the size of a record in bits, and $n$
is the total number of records being merged.  The cost of the merge is
bounded as follows:

\begin{lemma}\label{lem:merge}
  $l = \branchfactor{}M/B$ sorted sequences with total size $n$ (stored in
  $\lceil n/B \rceil$ blocks, and block aligned) can be merged using
  at most $(\branchfactor{} + 1) \lceil n/B \rceil$ reads and $\lceil n/B
  \rceil$ writes, on the AEM model with primary memory size
  $\left(M+2B+2\alpha\branchfactor M/B\right)$.
\end{lemma}

\begin{proof}
Each round (except perhaps the last) outputs at least $M$ records, and
hence the total number of rounds is at most $\lceil n/M \rceil$.  The
first phase of each round requires at most $\branchfactor{}M/B$ reads, so the
total number of reads across all the first phases is at most
$\branchfactor{}\lceil n/B \rceil$ (the last round can be included in this
since it only loads as many blocks as are output).  For the second
phase, a block is only read when incrementing its pointer, therefore
every block is only read once in the second phase.  Also every record is
only written once.  This gives the stated bounds on the number of
reads and writes.  The space includes the space for the in-memory heap
($M$), the load and store buffers, the pointers $I$ ($\alpha
\branchfactor M/B$), and pointers to maintain the last-record in block
information ($\alpha \branchfactor M/B$).
\end{proof}

We note that it is also possible to keep $I$ in secondary memory.
This will double the number of writes because every time the algorithm
moves to a new block in an input array $i$, it would need to write out
the updated $I_i$.  The increase in reads is small.  Also, if one uses
a balanced search tree to implement the priority queue $Q$ then
the size increases by $< M (\log M)/s$ in order to store the pointers
in the tree.

For the base case when $n \leq \branchfactor{}M <\wcost{}M$ we use the following lemma.

\begin{lemma}\label{lem:selsort}
  $n \leq \branchfactor{}M$ records stored in $\lceil n/B \rceil$ blocks can
  be sorted using at most $\branchfactor{} \lceil n/B
  \rceil$ reads and $\lceil n/B \rceil$ writes, on the AEM model with primary
  memory size $M + B$.
\end{lemma}
\begin{proof}
  We sort the elements using a variant of selection sort, scanning the
  input list a total of at most $\branchfactor{}$ times.  In the first scan,
  store in memory the $M$ smallest elements seen so far, performing no
  writes and $\lceil n/B \rceil$ reads.  After completing the scan,
  output all the $\min(M,n)$ elements in sorted order using
  $\lceil\min(M,n)/B\rceil$ writes.  Record the maximum element written so far.
  In each subsequent phase (if not finished), store in memory the $M$
  smallest records larger than the maximum written so far, then output
  as before.  The cost is $\lceil n/B \rceil$ reads and $M/B$ writes
  per phase (except perhaps the last phase).  We need one extra block
  to hold the input.  The largest output can be stored in the
  $O\left(\log M\right)$ locations we have allowed for in the model.
  This gives the stated bounds because every element is written out once
  and the input is scanned at most $\wcost$ times.
\end{proof}

Together we have:

\begin{theorem}\label{lem:mergesort}
  Algorithm~\ref{alg:mergesort} sorts $n$ records
  using
\[R\left(n\right) \leq \left(\branchfactor{}+1\right)\left\lceil \frac{n}{B}\right\rceil\left\lceil\log_{\frac{\branchfactor{}M}{B}}\left(\frac{n}{B}\right)\right\rceil\]
  reads, and
\[W\left(n\right) \leq  \left\lceil\frac{n}{B}\right\rceil\left\lceil\log_{\frac{\branchfactor{}M}{B}}\left(\frac{n}{B}\right)\right\rceil\]
writes on an AEM with primary memory size $\left(M+2B+2\alpha\branchfactor M/B\right)$. % with $\wcost \leq B$.
\end{theorem}

\begin{proof}
The number of recursive levels of merging is bounded by
$\left\lceil\log_{\frac{\branchfactor{}M}{B}}\left(\frac{n}{\branchfactor{}M}\right)\right\rceil$,
and when we add the additional base round we have $1 +
\left\lceil\log_{\frac{\branchfactor{}M}{B}}\left(\frac{n}{\branchfactor{}M}\right)\right\rceil
=
\left\lceil\log_{\frac{\branchfactor{}M}{B}}\left(\frac{n}{\branchfactor{}M}\frac{\branchfactor{}M}{B}\right)\right\rceil
=
\left\lceil\log_{\frac{\branchfactor{}M}{B}}\left(\frac{n}{B}\right)\right\rceil$.
The cost for each level is at most $(\branchfactor + 1) \lceil
n/B \rceil$ reads and $\lceil n/B \rceil$ writes (only one block on
each level might not be full).
\end{proof}

\begin{corollary}
Assuming that $n$ is big enough so we ignore the ceiling function in number of levels for merging, then picking $\branchfactor$ such that $\displaystyle{\branchfactor \over \log\branchfactor}<{\wcost\over \log{M \over B}}$ gives an overall improvement of the total I/O complexity (reads and writes).
\end{corollary}
The details of computing this can be found in appendix~\ref{app:branchfactor}.  In real-world applications for sorting, $\left\lceil\log_{\frac{M}{B}}\left(\frac{n}{B}\right)\right\rceil$ is usually a small constant $p$ (between $2$ to $6$). In this case, we can always try all $\branchfactor=\lceil\left(\frac{n}{B}\right)^{1/p'}/\frac{M}{B}\rceil$ for all integer $p'$ and $1\leq p'\leq p$, and pick the appropriate $\branchfactor$ that gives the minimum total I/O complexity.  Our new mergesort will never perform worse than the classic $M/B$-way mergesort, since we always have the option to pick $\branchfactor$ as $1$, and the new algorithm will performe exactly the same as the classic EM mergesort. 
\subsection{Sample Sort}
\label{sec:io-sample-sort}

We now describe an $l$-way randomized sample
sort~\cite{Blelloch91,Frazer70} (also called distribution sort), which
asymptotically matches the I/O bounds of the mergesort.  The idea of
sample sort is to partition $n$ records into $l$ approximately equally
sized buckets based on a sample of the keys within the records, and
then recurse on each bucket until an appropriately-sized base case is
reached.  Similar to the mergesort, here we use a branching factor
$l = \branchfactor{}M/B$, where $1\leq\branchfactor=\wcost$.  
Again this branching factor reduces the number
of levels of recursion relative to the standard EM sample sort which
uses $l = M/B$~\cite{AggarwalV88}, and the analysis of picking appropriate $\branchfactor$ in previous section also holds here.  We describe how to process each
partition and the base case.

The partitioning starts by selecting a set of splitters.  This can be
done using standard techniques, which we review later.  The splitters
partition the input into buckets that w.h.p.~are within
a constant factor of the average size $n/l$.  The algorithm now needs
to bucket the input based on the splitters.  The algorithm processes
the splitters in $\branchfactor{}$ rounds of size $M/B$ each, starting with
the first $M/B$ splitters.  For each round the algorithm scans the
whole input array, partitioning each value into the one of $M/B$
buckets associated with the splitters, or skipping a record if its key
does not belong in the current buckets.  One block for each bucket is
kept in memory.  Whenever a block for one of the buckets is full, it
is written out to memory and the next block is started for that
bucket.  Each $\branchfactor$ rounds reads all of the input and writes out
only the elements associated with these buckets (roughly a $1/\branchfactor$
fraction of the input).

The base case occurs when $n\leq \branchfactor{}M$, at which point we apply
the selection sort from Lemma~\ref{lem:selsort}.

Let $n_0$ be the original input size.
The splitters can be chosen by randomly picking a sample of keys of
size $m = \Theta(l \log n_0)$, sorting them, and then sub-selecting the
keys at positions $m/l, 2m/l, \ldots, (l-1)m/l$.  By selecting the
constant in the $\Theta$ sufficiently large, this process ensures
that, w.h.p., every bucket is within a constant factor
of the average size~\cite{Blelloch91}.  To sort the samples apply a
RAM mergesort, which requires at most
$O\left(\left(\left(l \log n_0\right)/B\right) \log \left(l \log
    n_0/M\right)\right)$
reads and writes.  This is a lower-order term when
$l = O(n\,/\log^2 n)$, but unfortunately this bound on $l$ may
not hold for small subproblems.  There is a simple solution---when
$n\leq \branchfactor{}^2M^2/B$, instead use $l=n/(\branchfactor{}M)$.  With this
modification, we always have $l \leq \sqrt{n/B}$.

It is likely that the splitters could also be selected
deterministically using an approach used in the original I/O-efficient
distribution sort~\cite{AggarwalV88}.

\begin{theorem}
  The $\branchfactor{}M/B$-way sample sort sorts $n$ records using, w.h.p.,
\[R\left(n\right) = O\left(\frac{\branchfactor{}n}{B} \left\lceil\log_{\frac{\branchfactor{}M}{B}}\left(\frac{n}{B}\right)\right\rceil\right)\]
  reads, and
\[W\left(n\right) = O\left(\frac{n}{B}\left\lceil\log_{\frac{\branchfactor{}M}{B}}\left(\frac{n}{B}\right)\right\rceil\right)\]
writes on an AEM with primary memory size $\left(M+B+M/B\right)$.
\end{theorem}
\begin{proof} (Sketch) The primary-memory size allows one block from
  each bucket as well as the $M/B$ splitters to remain in memory.  Each
  partitioning step thus requires $\lceil n/B \rceil + \branchfactor{}M/B$
  writes, where the second term arises from the fact that each bucket
  may use a partial block.  Since $n \geq \branchfactor{}M$ (this is not a
  base case), the cost of each partitioning step becomes $O(n/B)$
  writes and $O(\branchfactor{}n/B)$ reads.
  % INSERT SOMETHING ABOUT CHOOSING THE SPLITTERS. The following isn't
  % really quite right
  Because the number of splitters is at most
  $\sqrt{n} = O(n/\log^2n)$, choosing and sorting the splitters
  takes $O(n/B)$ reads and writes.
  Observe that the recursive structure matches that of a sample sort with
  an effective memory of size $\branchfactor{}M$, and that there will be at most two
  rounds at the end where $l = n/(\branchfactor M)$.  As in standard sample
  sort, the number of writes is linear with the size of the subproblem, but
  here the number of reads is multiplied by a factor of $\branchfactor{}$. The
  standard samplesort analysis thus applies, implying the bound
  stated.

  It remains only to consider the base case.  Because all buckets are
  approximately the same size, the total number of leaves is
  $O(n/B)$---during the recursion, a size $n > \branchfactor{}M$ problem is split into
  subproblems whose sizes are $\Omega(B)$.  Applying
  Lemma~\ref{lem:selsort} to all leaves, we get a cost of $O(\branchfactor n
  / B)$ reads and $O(n/B)$ writes for all base cases.
\end{proof}

\myparagraph{Extensions for the Private-Cache Model}
The above can be readily parallelized.  Here we outline the approach.
We assume that there are $p = n/M$ processors.  We use parallelism
both within each partition, and across the recursive partitions.
Within a partition we first find the $l$ splitters in parallel.
(As above, $l = \branchfactor M/B$ except for the at most two rounds prior
to the base case where $l = n/(\branchfactor M)$.)
This can be done on a sample that is a logarithmic factor
smaller than the partition size, using a less efficient sorting algorithm
such as parallel mergesort, and then sub-selecting $l$ splitters from
the sorted order.
This requires $O(\branchfactor(M/B + \log^2 n))$ time, where the second term
($O(\branchfactor\log^2 n)$) is the depth of the parallel mergesort, and the
first term is the work term $O((\branchfactor/B)((n/\log n)\log n)/P) =
O(\branchfactor M/B)$.

The algorithm groups the input into $n/(\branchfactor M)$ chunks of size
$\branchfactor M$ each.  As before we also group the splitters into $\branchfactor$
rounds of size $M/B$ each.  Now in parallel across all chunks and
across all rounds, partition the chunk based on the round.  We have
$n/(\branchfactor M) \times \branchfactor = n/M$ processors so we can do them all in
parallel.  Each will require $\branchfactor M$ reads and $M$ writes.  To
ensure that the chunks write their buckets to adjacent locations (so
that the output of each bucket is contiguous) we will need to do a
pass over the input to count the size of each bucket for each chunk,
followed by a prefix sum.  This can be done before processing the
chunks and is a lower-order term.  The time for the computation is
$O(\branchfactor M/B)$.

The processors are then divided among the sub-problems proportional to
the size of the sub-problem, and we repeat.  The work at each level of
recursion remains the same, so the time at each level remains the
same.  For the base case of size $\leq \branchfactor M$, instead of using a
selection sort across all keys, which is sequential, we find $\branchfactor$
splitters and divide the work among $\branchfactor$ processors to sub-select
their part of the input, each by reading the whole input, and then
sorting their part of size $O(M)$ using a selection sort on those
keys.  This again takes $O(\branchfactor M/B)$ time.  The total time for the
algorithm is therefore:
\[O\left(\branchfactor\left({M\over B} + \log^2 n\right) \left\lceil 1 + \log_{\frac{\branchfactor{}M}{B}}\left(\frac{n}{\branchfactor{}M}\right)\right\rceil\right)\]
with high probability.    This is linear speedup assuming $\frac{M}{B}
\geq \log^2 n$.  Otherwise the number of processors can be reduced to
maintain linear speedup.

\subsection{I/O Buffer Trees}\label{sec:io-buffer-trees}

This section describes how to augment the basic buffer
tree~\cite{Arge03} to build a priority queue that supports $n$
\textsc{Insert} and \textsc{Delete-Min} operations with an amortized
cost of $O((\branchfactor{}/B)(1+\log_{\branchfactor{}M/B} n))$ reads
and $O((1/B)(1+ \log_{\branchfactor{}M/B} n ))$ writes per
operation.  Using the priority queue to implement a sorting algorithm
trivially results in a sort costing a total of
$O((\branchfactor{}n/B)(1+ \log_{\branchfactor{}M/B} n))$ reads and $O((n/B)(1+
\log_{\branchfactor{}M/B} n))$ writes.
Here $\branchfactor$ is also the extra branching factor to reduce writes (at a cost of more reads), and can be chosen using similar argument in two previous sorting algorithms.
These
bounds asymptotically match the preceding sorting algorithms, but some
additional constant factors are introduced because a buffer tree is a
dynamic data structure.

Our buffer tree-based priority queue for the AEM contains a few differences from
the regular EM buffer tree~\cite{Arge03}: (1) the buffer tree nodes are
larger by a factor $\branchfactor$, (2) consequently, the ``buffer-emptying''
process uses an efficient sort on $\branchfactor{}M$ elements instead of an
in-memory sort on $M$ elements, and (3) to support the priority queue,
$O(\branchfactor{}M)$ elements are stored outside the buffer tree instead of $O(M)$,
which adds nontrivial changes to the data structure.

\subsubsection{Overview of a buffer tree}

A buffer tree~\cite{Arge03} is an augmented version of an
$(a,b)$-tree~\cite{HuddlestonMe82}, where $a=l/4$ and $b=l$ for large
branching factor $l$.  In the original buffer tree $l=M/B$, but to
reduce the number of writes we instead set $l=\branchfactor{}M/B$.  As an
$(a,b)$ tree, all leaves are at the same depth in the tree, and all
internal nodes have between $l/4$ and $l$ children (except the root,
which may have fewer).  Thus the height of the tree is $O(1+\log_l
n)$.  An internal node with $c$ children contains $c-1$ keys, stored
in sorted order, that partition the elements in the subtrees.  The
structure of a buffer tree differs from that of an $(a,b)$ tree in two
ways.  Firstly, each leaf of the buffer tree contains between $lB/4$
and $lB$ elements stored in $l$ blocks.\footnote{Arge~\cite{Arge03}
  defines the ``leaves'' of a buffer tree to contain $\Theta(B)$
  elements instead of $\Theta(lB)$ elements.  Since the algorithm only
  operates on the parents of those ``leaves'', we find the terminology
  more convenient when flattening the bottom two levels of the
  tree. Our leaves thus correspond to what Arge terms ``leaf
  nodes''~\cite{Arge03} (not to be confused with leaves) or
  equivalently what Sitchinava and Zeh call ``fringe
  nodes''~\cite{SitchinavaZe12}.}  Secondly, each node in the buffer
tree also contains a dense unsorted list, called a \defn{buffer}, of
partially inserted elements that belong in that subtree.

We next summarize the basic buffer tree insertion
process~\cite{Arge03}.  Supporting general deletions is not much
harder, but to implement a priority queue we only need to support
deleting an entire leaf. The insertion algorithm proceeds in two
phases: the first phase moves elements down the tree through buffers,
and the second phase performs the $(a,b)$-tree rebalance operations
(i.e., splitting nodes that are too big).  The first phase begins by
appending the new element to the end of the root's buffer. We say that
a node is \defn{full} if its buffer contains at least $lB$ elements.  If
the insert causes the root to become full, then a \defn{buffer-emptying
  process} commences, whereby all of the elements in the node's buffer
are sorted then distributed to the children (appended to the ends of
their buffers).  This distribution process may cause children to
become full, in which case they must also be emptied.  More precisely,
the algorithm maintains a list of internal nodes with full buffers
(initially the root) and a separate list of leaves with full buffers.
The first phase operates by repeatedly extracting a full internal node
from the list, emptying its buffer, and adding any full children to
the list of full internal or leaf nodes, until there are no full internal
nodes.

Note that during the first phase, the buffers of full nodes may far
exceed $lB$, e.g., if all of the ancestors' buffer elements are
distributed to a single descendant.  Sorting the buffer from scratch
would therefore be too expensive.  Fortunately, each distribution
process writes elements to the child buffers in sorted order, so all
elements after the $lB$'th element (i.e., those written in the most
recent emptying of the parent) are sorted. It thus suffices to split
the buffer at the $lB$'th element and sort the first $lB$ elements,
resulting in a buffer that consists of two sorted lists.  These two
lists can trivially be merged as they are being distributed to the
sorted list of children in a linear number of I/O's.

When the first phase completes, there may be full leaves but no full
internal nodes. Moreover, all ancestors of each full leaf have empty
buffers.  The second phase operates on each full leaf one at a time.
First, the buffer is sorted as above and then merged with the elements
stored in the leaf.  If the leaf contains $X > lB$ elements, then a
sequence of $(a,b)$-tree rebalance operations occur whereby the leaf
may be split into $\Theta(X/(lB))$ new nodes.  These splits cascade up
the tree as in a typical $(a,b)$-tree insert.

\subsubsection{Buffer tree with fewer writes}

To reduce the number of writes, we set the branching factor of the
buffer tree to $l = \branchfactor{}M/B$ instead of $l=M/B$.  The consequence of this
increase is that the buffer emptying process needs to sort $lB=\branchfactor{}M$
elements, which cannot be done with an in-memory sort.  The
advantage is that the height of the tree reduces to
$O(1+\log_{\branchfactor{}M/B} n)$.

\begin{lemma}\label{lem:flushcost}
  It costs $O(\branchfactor{}X/B)$ reads and $O(X/B)$ writes to empty a full buffer
  containing $X$ elements using $\Theta(M)$ memory.
\end{lemma}
\begin{proof}
  By Lemma~\ref{lem:selsort},
  the cost of sorting the first $\branchfactor{}M$ elements is $O(\branchfactor{}^2M/B)$ reads and
  $O(\branchfactor{}M/B)$ writes.  The distribute step can be performed by
  simultaneously scanning the sorted list of children along with the
  two sorted pieces of the buffer, and outputting to the end of the
  appropriate child buffer.  A write occurs only when either finishing
  with a child or closing out a block.  The distribute step thus uses
  $O(\branchfactor{}M/B + X/B)$ reads and writes, giving a total of $O(\branchfactor{}^2M/B +
  X/B)$ reads and $O(\branchfactor{}M/B + X/B)$ writes including the sort step.
  Observing that full means $X > \branchfactor{}M$ completes the proof.
\end{proof}

\begin{theorem}\label{thm:buffertree}
  Suppose that the partially empty block belonging to the root's
  buffer is kept in memory.  Then the amortized cost of each insert
  into an $n$-element buffer tree is
  $O((\branchfactor{}/B)(1+\log_{\branchfactor{}M/B} n))$ reads and
  $O((1/B)(1+\log_{\branchfactor{}M/B} n))$ writes.
\end{theorem}
\begin{proof}
  This proof follows from Arge's buffer tree performance
  proofs~\cite{Arge03}, augmented with the above lemma. We first
  consider the cost of reading and writing the buffers. The last block
  of the root buffer need only be written when it becomes full, at
  which point the next block must be read, giving $O(1/B)$ reads and
  writes per insert.  Each element moves through buffers on a
  root-to-leaf path, so it may belong to $O(1+\log_{\branchfactor{}M/B} n)$ emptying
  processes.  According to Lemma~\ref{lem:flushcost}, emptying a full
  buffer costs $O(\branchfactor{}/B)$ reads and $O(1/B)$ writes per element.
  Multiplying these two gives an amortized cost per element matching
  the theorem.

  We next consider the cost of rebalancing operations.  Given the
  choice of $(a,b)$-tree parameters, the total number of node splits
  is $O(n/(lB))$~\cite[Theorem 1]{Arge03} which is $O(n/(\branchfactor{}M))$.  Each
  split is performed by scanning a constant number of nodes, yielding
  a cost of $O(\branchfactor{}M/B)$ reads and write per split, or $O(n/(\branchfactor{}M)\cdot
  \branchfactor{}M/B) = O(n/B)$ reads and writes in total or $O(1/B)$ per insert.
\end{proof}

\subsubsection{An efficient priority queue with fewer writes}

The main idea of Arge's buffer tree-based priority queue~\cite{Arge03}
is to store a working set of the $O(lB)$ smallest elements resident in
memory.  When inserting an element, first add it to the working set,
then evict the largest element from the working set (perhaps the one
just inserted) and insert it into the buffer tree.  To extract the
minimum, find it in the working set.  If the working set is empty,
remove the $\Theta(lB)$ smallest elements from the buffer tree and add
them to the working set.  In the standard buffer tree, $l=M/B$ and
hence operating on the working set is free because it fits entirely in
memory.  In our case, however, extra care is necessary to
maintain a working set that has size roughly $\branchfactor$ times larger.

Our AEM priority queue follows the same idea except the working set is
partitioned into two pieces, the alpha working set and beta working
set.  The \defn{alpha working set}, which is always resident in
memory, contains at most $M/4$ of the smallest elements in the
priority queue.  The \defn{beta working set} contains at most
$2\branchfactor{}M$ of the next smallest elements in the data structure,
stored in $O(\branchfactor{}M/B)$ blocks.
The motivation for having a beta working set is that during
\textsc{Delete-Min} operations, emptying elements directly from the
buffer tree whenever the
alpha working set is empty would be too expensive---having a beta
working set to stage larger batches of such elements leads
to better amortized bounds.
Coping with the interaction
between the alpha working set, the beta working set, and the buffer
tree, is the main complexity of our priority queue.  The beta working
set does not fit in memory, but we keep a constant number of blocks
from the beta working set and the buffer tree (specifically, the last
block of the root buffer) in memory.
%% The motivation for using the
%% beta working set is that during deletions, emptying elements from the
%% buffer tree directly to memory is too expensive, so we instead empty
%% more elements to the beta working set first to obtain better amortized
%% bounds.

We begin with a high-level description of the priority-queue
operations, with details of the beta working set deferred until later.
For now, it suffices to know that we keep the maximum key in the beta
working set in memory.  To insert a new element, first compare its key
against the maximums in the alpha and beta working set. Then insert it
into either the alpha working set, the beta working set, or the buffer
tree depending on the key comparisons.  If the alpha working set
exceeds maximum capacity of $M/4$ elements, move the largest element
to the beta working set.  If the beta working set hits its maximum
capacity of $2\branchfactor{}M$ elements, remove the largest $\branchfactor{}M$
elements and insert them into the buffer tree.

To delete the minimum from the priority queue, remove the smallest
element from the alpha working set. If the alpha working set is empty,
extract the $M/4$ smallest elements from the beta working set (details
to follow) and move them to the alpha working set. If the beta
working set is empty, perform a buffer emptying process on the
root-to-leftmost-leaf path in the buffer tree. Then delete the
leftmost leaf and move its contents to the beta working set.

\myparagraph{The beta working set} The main challenge is in
implementing the beta working set. An unsorted list or buffer allows
for efficient inserts by appending to the last block. The challenge,
however, is to extract the $\Theta(M)$ smallest elements with $O(M/B)$
writes---if $\branchfactor{}>B$, each element may reside in a separate block,
and we thus cannot afford to update those blocks when extracting the
elements. Instead, we perform the deletions implicitly.

To facilitate implicit deletions, we maintain a list of ordered pairs
$(i_1,x_1),(i_2,x_2),(i_3,x_3),\ldots$, where $(i,x)$ indicates that
all elements with index at most $i$ and key at most $x$ are
invalid.
% For conciseness and without loss of generality, we assume
%that all keys in the beta working set are distinct (i.e., breaking
%ties by index).
Our algorithm maintains the invariant that for
consecutive list elements $(i_j,x_j)$ and $(i_{j+1},x_{j+1})$, we have
$i_j < i_{j+1}$ and $x_j > x_{j+1}$ (recall that all keys are distinct).

To insert an element to the beta working set, simply append it to the
end. The invariant is maintained because its index is larger than any
pair in the list.

To extract the minimum $M/4$ elements, scan from index $0$ to $i_1$ in
the beta working set, ignoring any elements with key at most $x_1$.
Then scan from $i_1+1$ to $i_2$, ignoring any element with key at most
$x_2$.  And so on.  While scanning, record in memory the $M/4$
smallest valid elements seen so far.  When finished, let $x$ be the
largest key and let $i$ be the length of the beta working set.  All
elements with key at most $x$ have been removed from the full beta
working set, so they should be implicitly marked as invalid. To
restore the invariant, truncate the list until the last pair
$(i_j,x_j)$ has $x_j > x$, then append $(i,x)$ to the list.  Because
the size of the beta working set is growing, $i_j < i$.  It should be
clear that truncation does not discard any information as $(i,x)$
subsumes any of the truncated pairs.

Whenever the beta working set grows too large ($2\branchfactor{}M$ valid elements)
or becomes too sparse ($\branchfactor$ extractions of $M/4$ elements each have
occurred), we first rebuild it.  Rebuilding scans the elements in
order, removing the invalid elements by packing the valid ones densely
into blocks.  Testing for validity is done as above.  When done, the
list of ordered pairs to test invalidity is cleared.

Finally, when the beta working set grows too large, we extract the
largest $\branchfactor{}M$ elements by sorting it (using the selection sort of
Lemma~\ref{lem:selsort}).

\myparagraph{Analyzing the priority queue}
We begin with some lemmas about the beta working set.

\begin{lemma}\label{lem:extract}
  Extracting the $M/4$ smallest valid elements from the beta working
  set and storing them in memory costs $O(\branchfactor{}M/B)$ reads and amortized
  $O(1)$ writes.
\end{lemma}
\begin{proof}
  The extraction involves first performing read-only passes over the
  beta working set and list of pairs, keeping one block from the
  working set and one pair in memory at a time.  Because the working
  set is rebuilt after $\branchfactor$ extractions, the list of pairs can
  have at most $\branchfactor$ entries.  Even if the list is not I/O
  efficient, the cost of scanning both is $O(\branchfactor{}M/B+\branchfactor) =
  O(\branchfactor{}M/B)$ reads.  Next the list of pairs indicating invalid
  elements is updated.  Appending one new entry requires $O(1)$ writes.
  Truncating and deleting any old entries can be charged against their
  insertions.
\end{proof}

The proof of the following lemma is similar to the preceding one, with
the only difference being that the valid elements must be moved and
written as they are read.

\begin{lemma}\label{lem:rebuild}
  Rebuilding the beta working set costs $O(\branchfactor{}M/B)$ reads and writes. \qed
\end{lemma}

\begin{theorem}
  Our priority queue, if initially empty, supports $n$ \textsc{Insert} and
  \textsc{Delete-Min} operations with an amortized cost of
  $O((\branchfactor{}/B)(1+\log_{\branchfactor{}M/B} n))$ reads and
  $O((1/B)(1+ \log_{\branchfactor{}M/B} n ))$ writes per
  operation.
\end{theorem}
\begin{proof}
  Inserts are the easier case.  Inserting into the alpha working set
  is free.  The amortized cost of inserting directly into the beta
  working set (a simple append) is $O(1/B)$ reads and writes, assuming
  the last block stays in memory.  The cost of inserting directly into
  the buffer tree matches the theorem.  Occasionally, the beta working
  set overflows, in which case we rebuild it, sort it, and insert
  elements into the buffer tree.  The rebuild costs $O(\branchfactor{}M/B)$
  reads and writes (Lemma~\ref{lem:rebuild}), the sort costs
  $O(\branchfactor{}^2M/B)$ reads and $O(\branchfactor{}M/B)$ writes (by
  Lemma~\ref{lem:selsort}), and the $\branchfactor{}M$ buffer tree inserts
  cost $O((\branchfactor{}^2M/B)(1+\log_{\branchfactor{}M/B}n))$ reads and
  $O((\branchfactor{}M/B)(1+\log_{\branchfactor{}M/B}n))$ writes (by
  Theorem~\ref{thm:buffertree}).  The latter dominates.  Amortizing
  against the $\branchfactor{}M$ inserts that occur between overflows, the
  amortized cost per insert matches the theorem statement.

  Deleting the minimum element from the alpha working set is free.
  When the alpha working set becomes empty, we extract $M/4$ elements
  from the beta working set, with a cost of $O(\branchfactor{}M/B)$ reads and
  $O(1)$ writes (Lemma~\ref{lem:extract}).  This cost may be
  amortized against the $M/4$ deletes that occur between extractions,
  for an amortized cost of $O(\branchfactor{}/B)$ reads and $O(1/M)$ writes
  per delete-min. Every $\branchfactor$ extractions of $M/4$ elements, the
  beta working set is rebuilt, with a cost of $O(\branchfactor{}M/B)$ reads
  and writes (Lemma~\ref{lem:rebuild}) or amortized $O(1/B)$ reads and
  writes per delete-min.  Adding these together, we so far have
  $O(\branchfactor{}/B)$ reads and $O(1/B)$ writes per delete-min.

  It remains to analyze the cost of refilling the beta working set
  when it becomes empty.  The cost of removing a leaf from the buffer
  tree is dominated by the cost of emptying buffers on a
  length-$O(\log_{\branchfactor{}M/B}n)$ path. Note that the buffers are not
  full, so we cannot apply Lemma~\ref{lem:flushcost}. But a similar
  analysis applies.  The cost per node is $O(\branchfactor{}^2M/B + X/B)$
  reads and $O(\branchfactor{}M/B + X/B)$ writes for an $X$-element
  buffer. As with Arge's version of the priority queue~\cite{Arge03},
  the $O(X/B)$ terms can be charged to the insertion of the $X$
  elements, so we are left with a cost of $O(\branchfactor{}^2M/B)$ read and
  $O(\branchfactor{}M/B)$ writes per buffer. Multiplying by
  $O(1+\log_{\branchfactor{}M/B}n)$ levels gives a cost of
  $O((\branchfactor{}^2M/B)(1+\log_{\branchfactor{}M/B}n))$ reads and
  $O((\branchfactor{}M/B)(1+\log_{\branchfactor{}M/B}n))$ writes.  Because
  each leaf contains at least $\branchfactor{}M/4$ elements, we can amortize
  this cost against at least $\branchfactor{}M/4$ deletions, giving a cost
  that matches the theorem.
\end{proof}

With this priority queue, sorting can be trivially implemented in
$O((\branchfactor{}n/B)(1+ \log_{\branchfactor{}M/B} n))$ reads and $O((n/B)(1+
\log_{\branchfactor{}M/B} n))$ writes, matching the bounds of the previous
sorting algorithms.

\section{Cache-Oblivious Parallel Algorithms}\label{sec:low-depth-algs}

In this section we present low-depth cache-oblivious parallel
algorithms for sorting and Fast Fourier Transform,
with asymmetric read and write costs.  Both algorithms
(i) have only polylogarithmic depth, (ii) are processor-oblivious
(i.e., no explicit mention of processors), (iii) can be
cache-oblivious or cache-aware, and (iv) map to low cache complexity
on parallel machines with hierarchies of shared caches as well as
private caches using the results of Section~\ref{sec:prelims}.
We also present a linear-depth, cache-oblivious parallel algorithm for matrix
multiplication.  All three algorithms use $\Theta(\wcost)$ fewer writes
than reads.

Our algorithms are oblivious to cache size.  However, we still assume that
the read/write ratio $\wcost{}$ is aware since it is a parameter for the
main memory instead of the caches.

% Sorting figure put here for better page placement
\begin{figure*}[!t]\centering
\includegraphics[width=\linewidth]{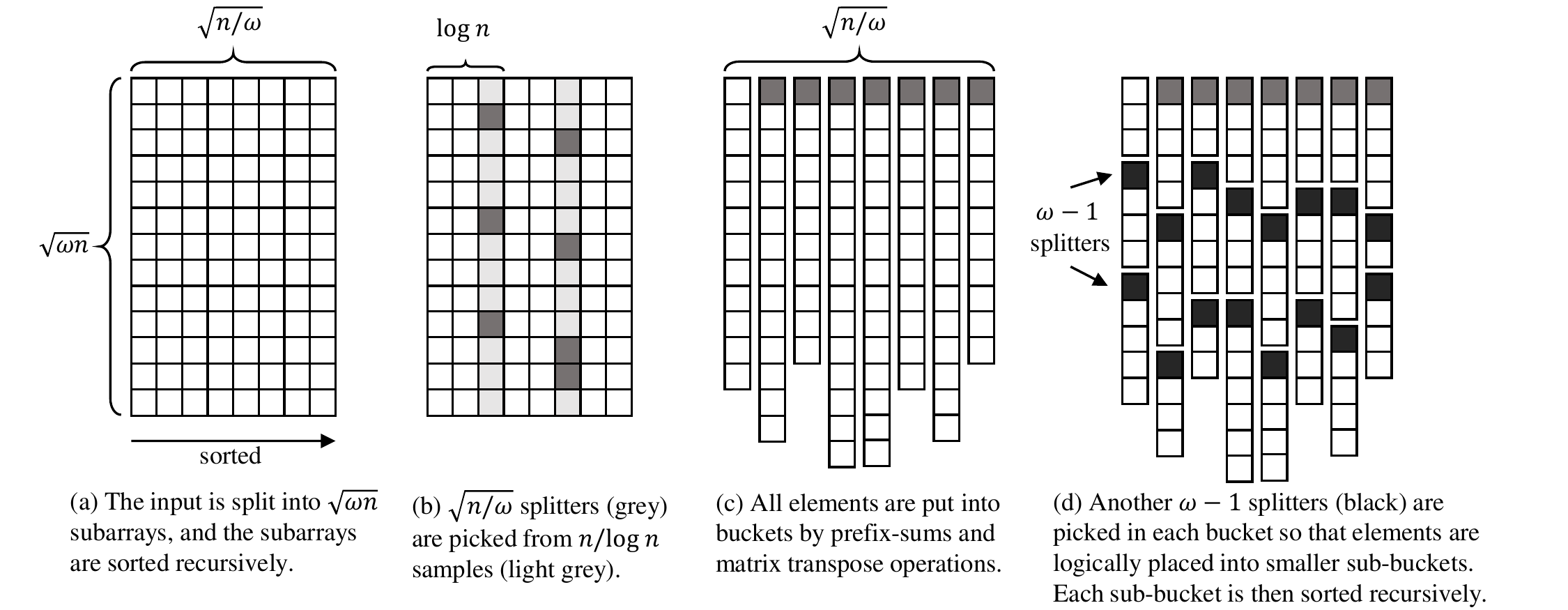}
\vspace{-1em}
\caption{The low-depth cache-oblivious algorithm on asymmetric read
  and write costs to sort an input array of size
  $n$. }\label{fig:cache-obli-sort}
\end{figure*}

\subsection{Sorting}\label{sec:low-depth-sort}

\newcommand{\wdepth}{\wcost}

We show how the low-depth, cache-oblivious sorting algorithm
from~\cite{BGS10} can be adapted to the asymmetric case.  The original
algorithm is based on viewing the input as a $\sqrt{n} \times
\sqrt{n}$ array, sorting the rows, partitioning them based on
splitters, transposing the partitions, and then sorting the buckets.
The original algorithm incurs $O\left(\left(n/B\right)\log_M\left(n\right)\right)$ reads and writes.
To reduce the number of writes,
our revised version partitions slightly differently and does extra
reads to reduce the number of levels of recursion.  The algorithm
does $O\left(\left(n/B\right) \log_{\wcost M}\left(\wcost n\right)\right)$ writes,
$O\left(\left(\wcost n/B\right) \log_{\wcost M}\left(\wcost n\right)\right)$ reads, and has depth
$O\left(\wdepth \log^2 \left(n/\wcost\right)\right)$ w.h.p.

The algorithm uses matrix transpose, prefix sums and mergesort as
subroutines.  Efficient parallel and cache-oblivious versions of these
algorithm are described in~\cite{BGS10}.  For an input of size $n$,
prefix sums has depth $O\left(\wdepth \log n\right)$ and requires $O\left(n/B\right)$ reads
and writes, merging two arrays of lengths $n$ and $m$ has depth
$O\left(\wdepth\log\left(n+m\right)\right)$ and requires $O\left(\left(n+m\right)/B\right)$ reads and writes, and
mergesort has depth $O\left(\wdepth \log^2 n\right)$ and requires
$O\left((n/B)\log_2\left(n/M\right)\right)$ reads and writes.  Transposing an $n\times m$
matrix has depth $O\left(\wdepth \log\left(n+m\right)\right)$ and requires $O\left(nm/B\right)$ reads
and writes.

Our cache-oblivious sorting algorithm works recursively, with a base
case of $n \leq M$, at which point any parallel sorting algorithm with
$O\left(n\log n\right)$ reads/writes and $O\left(\wdepth \log n\right)$ depth can be applied (e.g.~\cite{Cole1988}).
%if we assume $k<M$.

Figure~\ref{fig:cache-obli-sort} illustrates the steps of the algorithm.
Given an input array of size $n$, the algorithm first splits it into
$\sqrt{n\wcost}$ subarrays of size $\sqrt{n/\wcost}$ and recursively
sorts each of the subarrays.  This step corresponds to step (a) in
Figure~\ref{fig:cache-obli-sort}.

Then the algorithm determines the splitters by sampling.  After the
subarrays are sorted, every $\left(\log n\right)$'th element from each
row is sampled, and these $n/\log n$ samples are sorted using a
cache-oblivious mergesort.  Then $\sqrt{n/ \wcost}-1$
evenly-distributed splitters are picked from the sorted samples to
create $\sqrt{n/\wcost}$ buckets.  The algorithm then determines the
boundaries of the buckets in each subarray, which can be implemented
by merging the splitters with each row, requiring
$O\left(\wdepth\log n\right)$ depth and $O\left(n/B\right)$ writes
overall.  This step is shown as step (b) in
Figure~\ref{fig:cache-obli-sort}.  Notice that on average the size of
each bucket is $O\left(\sqrt{n\wcost}\right)$, and the largest bucket
has no more than $2\sqrt{n\wcost}\log n$ elements.

After the subarrays are split into $\sqrt{n/\wcost}$ buckets, prefix
sums and a matrix transpose can be used to place all keys destined for
the same bucket together in contiguous arrays.  This process is
illustrated as step (c) in Figure~\ref{fig:cache-obli-sort}.  This
process requires $O\left(n/B\right)$ reads and writes, and $O\left(\wdepth \log n\right)$
depth.

The next step is new to the asymmetric algorithm and is the part that
requires extra reads.  As illustrated in
Figure~\ref{fig:cache-obli-sort} (d), $\wcost-1$ pivots are chosen from
each bucket to generate $\wcost$ sub-buckets.
%If the size of one bucket is
%less than $k$, it reduces to the base case that we have
%discussed.
We sample $\max\{\wcost,\sqrt{\wcost n}/\log n\}$ samples, apply a
mergesort, and evenly pick $\wcost-1$ pivots in the sample. Then the
size of each sub-bucket can be shown to be $O\left(\sqrt{n/\wcost}\log n\right)$
w.h.p. using Chernoff bounds.
%% $\Theta(\sqrt{n/\wcost})$ with high probability if $n\geq
%% \wcost\cdot \log^4 n$.
%% For smaller $n$, we can guarantee that the size
%% of a sub-bucket is $O(\sqrt{n/\wcost}\log n)$
%% with high probability.
%% , which indicates that resampling and
%% regenerating pivots can terminate in a constant number of round if
%% there exists a sub-bucket whose size exceeds the limit.
We then scan each bucket for $\wcost$ rounds to partition all elements
into $\wcost$ sub-buckets, and sort each sub-bucket recursively.

\begin{theorem}\label{thm:cache-obli-sort}
Our cache-oblivious sorting algorithm requires
  $O\left(\left(\wcost n/B\right)\log_{\wcost M}{\left(\wcost n\right)}\right)$ reads, 
  $O\left(\left(n/B\right)\log_{\wcost M}{\left(\wcost n\right)}\right)$ writes,
  and $O\left(\wdepth \log^2 \left(n/\wcost\right)\right)$ depth w.h.p.
\end{theorem}
\begin{proof}
All the subroutines except for the recursive calls do
$O\left(n/B\right)$ writes. The last partitioning process to put
elements into sub-buckets takes $O\left(\wcost n/B\right)$ reads and
the other subroutines require fewer reads. The overall depth is
dominated by the mergesort to find the first $\sqrt{n/\wcost}$ pivots,
requiring $O\left(\wcost\log^2 \left(n/\wcost\right)\right)$ depth per
level of recursion. Hence, the recurrence relations (w.h.p.) for read I/O's
($R$), write I/O's ($W$), and depth ($D$) are:
\vspace{-0.1in}
$$R\left(n\right)=O\left(\wcost n\over B\right)+\sqrt{\wcost n}\cdot R\left(\sqrt{n\over \wcost}\right)+\sum_{i=1}^{\sqrt{\wcost n}}R\left(n_i\right)$$
$$W\left(n\right)=O\left(n\over B\right)+\sqrt{\wcost n}\cdot W\left(\sqrt{n\over \wcost}\right)+\sum_{i=1}^{\sqrt{\wcost n}}W\left(n_i\right)$$
$$D\left(n\right)=O\left(\wcost\log^2 {n\over \wcost}\right)+\max_{i=1,\cdots,\sqrt{\wcost n}}\{D\left(n_i\right)\}$$
where $n_i$ is the size of the $i$'th sub-bucket, and $n_i= O\left(\sqrt{n/\wcost}\log n\right)$ and $\sum{n_i}=n$. The base case for the I/O complexity is $R\left(M\right)=W\left(M\right)=O\left(M/B\right)$. Solving these recurrences proves the theorem.
\end{proof}

\subsection{Fast Fourier Transform}\label{sec:fft}

We now consider a parallel cache-oblivious algorithm for computing the
Discrete Fourier Transform (DFT).  The algorithm we consider is based
on the Cooley-Tukey FFT algorithm~\cite{FFT}, and our description
follows that of~\cite{Frigo99}.  We assume that $n$ at each level of
recursion and $\wcost$ are powers of 2.  The standard cache-oblivious
divide-and-conquer algorithm~\cite{Frigo99} views the input matrix as
an $\sqrt{n} \times \sqrt{n}$ matrix, and incurs
$O\left(\left(n/B\right)\log_M\left(n\right)\right)$ reads and writes.
To reduce the number of writes, we modify the algorithm to the
following:

\begin{enumerate}\itemsep=0in
\item View input of size $n$ as a $\wcost\sqrt{n/\wcost}
  \times \sqrt{n/\wcost}$ matrix in row-major order.   Transpose
  the matrix.
\item Compute a DFT on each row of the matrix as follows
  \begin{enumerate}
  \item View the row as a $\wcost \times \sqrt{n/\wcost}$ matrix
  \item For each row $i$
    \begin{enumerate}
       \item Calculate the values of column DFTs for row $i$ using the
         brute-force method.  This will require $\wcost$ reads (each
         row) and $1$ write (row $i$) per value.
       \item Recursively compute the FFT for the row.
    \end{enumerate}
    \item Incorporate twiddle factors
    \item Transpose $\wcost \times \sqrt{n/\wcost}$ matrix
  \end{enumerate}
\item Transpose matrix
\item Recursively compute an FFT on each length $\sqrt{n/\wcost}$ row.
\item Transpose to final order.
\end{enumerate}

The difference from the standard cache-oblivious algorithms is the
extra level of nesting in step 2, and the use of a work-inefficient
DFT over $\wcost$ elements in step 2(b).  The transposes in steps 1,
2(d) and 3 can be done with $O\left(n/B\right)$
reads/writes and $O\left(\wcost\log n\right)$ depth.
% In fact the number of transposes is no more than for the standard version.
The brute-force calculations in step 2(b)i require a total of $\wcost n$ reads
(and arithmetic operations) and $n$ writes.   This is where we waste a
$\wcost$ factor in reads in order to reduce the recursion depth.  The
twiddle factors can all be calculated with $O(n)$ reads and writes.
There are a total of $2 \wcost \sqrt{n/\wcost}$ recursive calls on
problems of size $\sqrt{n/\wcost}$.

Our base case is of size $M$, which uses $M/B$ reads and writes.  The
number of reads therefore satisfies the following recurrence:
%\vspace{-0.1in}
\[
R\left(n\right) =
\begin{cases}
O\left(n/B\right) & \text{if $n \le M$} \\
2\wcost\sqrt{n/\wcost}\;R\left(\sqrt{n/\wcost}\right) + O\left(\wcost{}n/B\right) & \text{otherwise}
\end{cases}
\]
which solves to $R\left(n\right) = O\left(\left(kn/B\right)\log_{\wcost M}\left(\wcost n\right)\right)$, and the number of writes is
%\vspace{-0.1in}
\[
W\left(n\right) =
\begin{cases}
O\left(n/B\right) & \text{if $n \le M$} \\
2\wcost\sqrt{n/\wcost}\; W\left(\sqrt{n/\wcost}\right) + O\left(n/B\right) & \text{otherwise}
\end{cases}
\]
which solves to $W\left(n\right) = O\left(\left(n/B\right)\log_{\wcost M}\left(\wcost n\right)\right)$.

Since we can do all the rows in parallel for each step, and
the brute-force calculation in parallel, the only important dependence is that
we have to do step 2 before step 5.    This gives a recurrence
$D\left(n\right) =
2D\left(\sqrt{n/\wcost}\right) + O\left(\wcost\log n\right)$ for the depth, which
solves to $O\left(\wcost\log n\log\log n\right)$.

We note that the algorithm as described requires an extra transpose
and an extra write in step 2(b)i relative to the standard version.
This might negate any advantage from reducing the number of levels,
however we note that these can likely be removed.  In particular
the transpose in steps 2(d) and 3 can be merged by viewing the results
as a three dimensional array and transposing the first and last
dimensions (this is what the pair of transposes does).  The write in
step 2(b)i can be merged with the transpose in step 1 by
combining the columnwise FFT with the transpose, and applying it
$\wcost$ times.

%% We note that an FFT on a sequence of
%% length $2^d$ is a calculation on a butterlfy network starting with the
%% highest dimension, and going down the dimensions.

%% For simplicity assume that $k = 2^l$ and that $n = 2^(2r+l)$---there are a total of $2r + l$ dimensions (which needs to be true at each level of recursion).  We create a matrix with $2^r = \sqrt{n/k}$ rows each of length $2^{r+l} = \sqrt{kn}$.

%% \begin{enumerate}
%% \item Transpose the array.
%% \item Do a recursive FFT on each row (processing the top $r$ dimensions).
%% \item Transpose the results back.
%% \item For each row, and for each chunk within the row of size $2^r$ ($2^l$ of them per row):
%% \begin{enumerate}
%% \item scan the whole row of length $2^{r + l}$ to process the middle $l$ dimensions
%% writing the result into $2^r$ locations.
%% \item Do a recusive FFT on the chunk for the final $r$ dimensions.
%% \end{enumerate}
%% \end{enumerate}

%% This gives a recurrence for writes of the form $W(n) = 2k\sqrt{n/k}W(\sqrt{n/k}) + n/B$
%% with base case $W(kM) = kM$.    This solves to $W(n) = n/B \log_{k^2M}(kn)$.  This is
%% equivalent to $\frac{n}{B} \frac{\log n + \log k}{\log M + 2 \log k})$.
%% We have a factor of $k$ more reads.

\subsection{Matrix Multiplication}\label{sec:matmul}

In this section, we consider matrix multiplication in the asymmetric
read/write setting. The standard cubic-work sequential algorithm
trivially uses $O\left(n^3\right)$ reads and $\Theta\left(n^2\right)$
writes, one for each entry in the output matrix. For the EM model, the
blocked algorithm that divides the matrix into sub-matrices of size
$\sqrt{M} \times \sqrt{M}$ uses $O\left(n^3/(B\sqrt{M})\right)$
reads~\cite{Blumofe96,Frigo99}.  Because we can keep each $\sqrt{M}
\times \sqrt{M}$ sub-matrix of the output matrix in memory until it is
completely computed, the number of writes is proportional to the
number of blocks in the output, which is $O\left(n^2/B\right)$. This
gives the following simple result:

\begin{theorem}
External-memory matrix multiplication can be done in $O\left(n^3/(B\sqrt{M})\right)$
reads and $O\left(n^2/B\right)$ writes.
\end{theorem}

We now turn to a divide-and-conquer algorithm for matrix
multiplication, which is parallel and cache-oblivious.  We assume that
we can fit $3M$ elements in cache instead of $M$, which only affects
our bounds by constant factors.  Note that the standard
cache-oblivious divide-and-conquer algorithm~\cite{Blumofe96,Frigo99}
recurses on four parallel subproblems of size $n/2 \times n/2$, resulting in
$\Theta\left(n^3/(B\sqrt{M})\right)$ reads and writes.  To reduce the
writes by a factor of $\Theta\left(\wcost\right)$, we instead recurse
on $\wcost^2$ parallel subproblems (blocks) of size $n/\wcost \times n/\wcost$.
On each level of
recursion, computing each output block of size $n/\wcost \times
n/\wcost$ requires summing over the $\wcost$ products of
$n/\wcost\times n/\wcost$ size input matrices.  We assume the
recursive calls writing to the same target block are processed
sequentially so that all writes (and reads) can be made at the leaves
of the recursion to their final locations.

We now partition a task into $\wcost{}\times \wcost{}$ subproblems,
which may lead to up to $n\wcost{}/\sqrt{M} \times n\wcost{}/\sqrt{M}$
base cases and thus match the overall work of the standard cache-oblivious
algorithm.
Instead, in the first round, we pick an integer $b$ uniformly at random
from 1 to $\lfloor \log_2\wcost{}\rfloor$, and recurses on $b \times b$
subproblems.
After this round, we run the divide-and-conquer algorithms just described before with the branching factor of $\wcost{}\times \wcost{}$.
We show that this extra first step reduces $O(\log \wcost{})$ in exception.

Now we analyze the cost of this algorithm.
The base case is when the problem size $n$ is no more than $\wcost\sqrt{M}$.
At this point each
of its subproblems of size no more than $\sqrt{M} \times \sqrt{M}$ is writing into
an output block of size at most $M$, which fits in cache.  Therefore since we
do the products for the output blocks sequentially, the result can
stay in cache and only be written when all $\wcost$ are done.  The
number of writes is therefore $n^2/B$ total for $n\le\wcost\sqrt{M}$.
For reads all of the $\wcost^3$ subproblems might require reading
both inputs so there are $2\wcost{} n^2/B$ reads.  The non-base recursive
calls do not contribute any significant reads or writes since all
reading and writing to the arrays is done at the leaves.  This gives
us the following recurrence for the number of writes for an $n \times n$
matrix:
%\vspace{-0.05in}
\[
W\left(n\right) =
\begin{cases}
n^2/B & \text{if $n \leq \wcost\sqrt{M}$} \\
O(1) + {1\over \lfloor \log_2\wcost{}\rfloor}\sum_{k=1}^{\lfloor \log_2\wcost{}\rfloor}{(2^k)^3W(n/2^k)} & \text{the first round}\\
O(1) + \wcost^3W\left(n/\wcost\right) & \text{otherwise}
\end{cases}
\]

This solves to $W\left(n\right) = O\left(n^3/(
B\sqrt{M}\log\wcost)\right)$ in expectation, which is a factor of $\wcost$ less than for the
standard EM or cache-oblivious matrix multiply.
%\footnote{Note
%that for this analysis we assume the initial problem is
%  of size $n = \wcost^i\sqrt{M}$ for some integer $i$.}
The number of
reads satisfies:
%\vspace{-0.05in}
\[
R\left(n\right) =
\begin{cases}
2 \wcost n^2/B & \text{if $n \leq \wcost\sqrt{M}$} \\
O(1) + {1\over \lfloor \log_2\wcost{}\rfloor}\sum_{k=1}^{\lfloor \log_2\wcost{}\rfloor}{(2^k)^3R(n/2^k)} & \text{the first round}\\
O(1) + \wcost^3R\left(n/\wcost\right) & \text{otherwise}
\end{cases}
\]
This solves to $R\left(n\right) = O\left(n^3\wcost/(B\sqrt{M}\log \wcost)\right)$ in expectation.
%which is the same as for the standard EM or cache-oblivious matrix multiply.

Because the products contributing to a block are done sequentially,
the depth of the algorithm satisfies the recurrence $D(n) = \wcost D(n/\wcost) + O(1)$
with base case $D(1) = \wcost$,
which solves to $O(\wcost n)$.
This gives the following theorem:

\begin{theorem}
Our cache-oblivious matrix multiplication algorithm requires expected
$O\left(n^3\wcost/(B\sqrt{M}\log\wcost)\right)$ reads and
$O\left(n^3/(B\sqrt{M}\log\wcost)\right)$ writes, and has
$O\left(\wcost n\right)$ depth.
\end{theorem}

Comparing with the standard cache-oblivious matrix multiplication algorithm,
the overall read and write cost of new algorithm expects to be reduced by a factor of $O(\log\wcost{})$.

\section{Conclusion}\label{sec:conclusion}

Motivated by the high cost of writing relative to reading in emerging
non-volatile memory technologies, we have considered a variety of
models for accounting for read-write asymmetry, and proposed and
analyzed a variety of sorting algorithms in these models.  For the
asymmetric RAM and PRAM models, we have shown how to reduce the number
of writes from $O(n \log n)$ to $O(n)$ without asymptotically
increasing the other costs (reads, parallel depth).  For the
asymmetric external memory models (including the cache-oblivious
model) the reductions in writes are mostly more modest, often
increasing the base of a logarithm, and at the cost of asymptotically
more reads.  However, these algorithms might still have practical
benefit.  We also presented a cache-oblivious matrix-multiply that
asymptotically reduces the overall cost by a factor of $O(\log\wcost)$ in expectation.
Future work includes proving lower
bounds for the asymmetric external memory models, devising
write-efficient algorithms for additional problems, and performing
experimental studies.

\section*{Acknowledgments}
This research was supported in part by NSF grants CCF-1218188,
CCF-1314633, and CCF-1314590,
and by the Intel Science and Technology Center for Cloud Computing.

%\bibliographystyle{abbrv}
%\bibliography{ref}

\balance

\appendix
\section{Appropriate Range of the Branching Factor $\branchfactor$}
\label{app:branchfactor}

In our new Asymmetric External Memory mergesort, we want to pick an appropriate branching factor $\branchfactor$ to ensure a better performance comparing to the classic mergesort on External Memory.

Based on Theorem~\ref{lem:mergesort}, the I/O complexity of Algorithm~\ref{alg:mergesort} is %$R(n)+\wcost{}W(n)$, which can be more specific as:
$$R(n)+\wcost{}W(n)=(\wcost{}+\branchfactor+1)\left\lceil \frac{n}{B}\right\rceil\left\lceil\log_{\frac{\branchfactor{}M}{B}}\left(\frac{n}{B}\right)\right\rceil$$ The I/O complexity for the classic mergesort is $R(n)+\wcost{}W(n)=(\wcost{}+1)\left\lceil \frac{n}{B}\right\rceil\left\lceil\log_{\frac{M}{B}}\left(\frac{n}{B}\right)\right\rceil$.
Assume $n$ is big enough so we ignore the ceiling function in number of levels for merging, then the following inequality should hold: 
$$(\wcost{}+\branchfactor+1)\left\lceil \frac{n}{B}\right\rceil\log_{\frac{\branchfactor{}M}{B}}\left(\frac{n}{B}\right)<(\wcost{}+1)\left\lceil \frac{n}{B}\right\rceil\log_{\frac{M}{B}}\left(\frac{n}{B}\right)$$ 
which solves to $\displaystyle{\branchfactor \over \log\branchfactor}<{\wcost+1\over \log{M \over B}}$, or $\displaystyle{\branchfactor \over \log\branchfactor}<{\wcost\over \log{M \over B}}$ if we assume $\wcost \gg 1$. If we plug in parameters $\wcost{}$, $M$ and $B$ of real-world hardwares, then any integer $\branchfactor$ s.t. $1\leq \branchfactor\leq 0.3\wcost{}$ satisfies the inequality, and decreases the I/O complexity of mergesort. 
\end{document}